\documentclass[journal]{IEEEtran}
\usepackage{subfigure,graphicx}
\ifCLASSINFOpdf
\else
\fi
\usepackage{amsthm,amsmath,amssymb}
\allowdisplaybreaks
\newtheorem{theorem}{\bf Theorem}

\newtheorem{corollary}{\bf Corollary}

\newtheorem{proposition}{\bf Proposition}

\usepackage{booktabs}
\usepackage{multirow,multicol}
\usepackage{xcolor}
\usepackage{algorithm,algpseudocode}
\usepackage{verbatim}
\usepackage{url}

\begin{document}
%
\title{Data-Driven Contract Design for Multi-Agent Systems with Collusion Detection}
%
%
%

\author{Nayara~Aguiar, Parv~Venkitasubramaniam and~Vijay~Gupta
\thanks{Nayara Aguiar and Vijay Gupta are with the Department of Electrical Engineering, University of Notre Dame, IN, USA
        {\tt\small \{ngomesde,vgupta2\}@nd.edu}}
\thanks{Parv Vekitasubramaniam is with the Department of Electrical Engineering, Lehigh University, PA, USA 
        {\tt\small pav309@lehigh.edu}}
}

\maketitle

\begin{abstract}
In applications such as participatory sensing and crowd sensing, self-interested agents exert costly effort towards achieving an objective for the system operator. We study such a setup where a principal incentivizes multiple agents of different types who can collude with each other to derive rent. The principal cannot observe the efforts exerted directly, but only the outcome of the task, which is a noisy function of the effort. The type of each agent influences the effort cost and task output. For a duopoly in which agents are coupled in their payments, we show that if the principal and the agents interact finitely many times, the agents can derive rent by colluding even if the principal knows the types of the agents. 
However, if the principal and the agents interact infinitely often, the principal can disincentivize agent collusion through a suitable data-driven contract. 
\end{abstract}


%
\IEEEpeerreviewmaketitle

\section{Introduction}
Many scenarios in smart infrastructure systems require a system operator to incentivize self-interested agents to exert costly effort to make decisions that align with the goal of the operator. For instance, in participatory sensing, a system operator requires many autonomous sensors to take measurements to allow estimation of a global quantity. The operator cannot observe directly the effort of each sensing agent (possibly for privacy reasons) and the agent might not benefit directly from the goal of the operator and thus needs to be compensated based on noisy outputs. Incentive design in this setting has been analyzed, particularly in the context of participatory sensing and crowd sensing, for classical estimation as well as learning based tasks (see~\cite{gao,restuccia} and the references therein). 

One possible abstraction of the problem is as a Stackelberg game (e.g. as adopted in the context of jamming \cite{garnaev,mukherjee}). A principal seeks to incentivize self-interested rational agents to exert costly effort to perform a task. The principal acting as the leader issues a contract that specifies the payment structure. The agents as the followers respond by deciding whether to accept the contract or not. If the contract is accepted, the agents further decide how much effort to exert. If the agents and the principal interact multiple times, then the entire process is repeated. In this paper, we focus on a duopoly, where each agent $i \in \mathcal{N}:=\{1,2\}$ can be defined as a type $\theta_i$ which influences his utility and the quality of his task output. Further, we assume that the principal can observe only the task output, but not the effort exerted by each agent, leading to a problem of moral hazard. 

Single-agent, single-principal problems with such (or similar) abstractions have been widely explored in the literature. Most such works that consider multiple agents being present assume that either the agents are decoupled or that collusion is prohibited, thus reducing the analysis to pairwise interactions between each individual agent and the principal (possibly with an additional optimization problem at the principal on which agents to select for the contract). We focus on the effect of collusion, specifically seeking to answer the question if it is possible for the agents to derive rent through collusion. We show that there is a difference between the principal and the agents interacting finitely or infinitely many times. Rent is possible in the former situation, but the principal can propose a dynamic contract with a collusion detection mechanism disincentivizing collusion in the latter. Interestingly, this result holds even if the principal knows the types of the agents. 

Three lines of literature are particularly relevant here. Dynamic contract design for multi-agent systems with information assymetry has been studied in game theory (see, e.g.,~\cite{holmstrom1999managerial,prat2014dynamic}. In signal processing applications, a repeated game is analyzed in \cite{dobakhshari2}, where the principal performs random verification and assigns a reputation score to each sensor. 
A dynamic estimation problem is studied in \cite{farokhi_dynamic}, leading to an equilibrium in which strategic sensors use a memory-less policy. In \cite{venkitasubramaniam}, a data-driven contract is analyzed for the situation when 
the quality of the output produced by participating agents depends on their true types. As opposed to these works, we consider multiple agents that are coupled both due to their compensation and possibility of collusion. 

Another line of relevant literature is that of payment schemes which couple agents. Competition among data aggregators is studied in \cite{westenbroek}, where aggregators are coupled in their costs and data sources are coupled in their compensation. Payment coupling is also considered in \cite{dobakhshari} in a 
demand response problem. However, the issue of collusion has received much less attention in the literature.

The third line of literature is one which studies collusion. In \cite{farokhi}, the formation of a coalition is analyzed for a class of static estimation problems with strategic sensors. In a static setting, the principal does not have information about whether agents are communicating to devise a joint strategy, which 
may lead to agents being able to derive rent by colluding. 
In contrast, our work considers the problem when the agents and the principal interact repeatedly and have the ability to change their strategies. We show that, in the limit of infinitely many interactions, the principal can utilize the previous  task outputs to perform a hypothesis test on whether collusion is underway and update the contract offered to disincentivize it. 

The remainder of this paper is organized as follows. Section~\ref{sec:formulation} introduces the problem formulation proposed. The static problem is solved in Section~\ref{sec:static}, where we derive conditions for collusion to benefit the agents. Section~\ref{sec:dynamic} presents the dynamic contract which enables the principal to verify if collusion is present. Lastly, Section~\ref{sec:conclusion} concludes this work.

\section{Problem Formulation}\label{sec:formulation}
Consider a Stackelberg game in which a principal, as a leader, establishes a contract to incentivize two agents to exert costly effort to perform a task. The task output of each agent $i \in \mathcal{N}:=\{1,2\}$ is represented by the random variable $X_i$, which can be observed by the principal. 
We let $X_i$ be defined by a function $f_i(a_i,\theta_i,\eta_i):\mathcal{A}\times \Theta\times \mathbb{R} \mapsto \mathbb{R}$, which depends on the agent's effort $a_i$ and type $\theta_i$, and on the random noise $\eta_i$. The noise variables $\{\eta_i\}$ are assumed to be independent and identically
distributed across the agents, and  bounded in the interval $\eta_i\in[\underline{\eta}_i,\overline{\eta}_i]$. 

The timeline for a specified time horizon $K$ is as follows. Initialize $k=0$.
\begin{enumerate}
    \item $t=4k+1$: The principal issues a contract defined by the payment $w_i(\mathbf{X})$, where $\mathbf{X}:=\{X_1,X_2\}$. 
    \item $t=4k+2$: Each agent decides whether or not to (i) accept this contract, and (ii) collude with the other agent. 
    \item $t=4k+3$: Non-colluding agents exert an effort $a_i$ which maximizes their individual utility, and colluding agents exert the effort $a_i$ that increases both of their utilities as compared to the non-colluding solution.
    \item $t=4k+4$: For each participating agent, the effort $a_i$ produces an output $X_i$, which is compensated by the principal according to $w_i(\mathbf{X})$.
    \item $k=k+1$ and repeat steps 1-5 while $k<K$.
\end{enumerate}

In this problem, the utility of the principal is given by
\begin{equation}\label{eq:U}
    U = \mathbb{E}\left[S(\mathbf{X})-\sum_{i=1}^2w_i(\mathbf{X})\right],
\end{equation}
where $S(\mathbf{X})=X_1+X_2$ is the benefit derived due to the outputs produced, and the expectation is taken over the output noise. The utility of each agent is defined as
\begin{equation}
    u_i = \mathbb{E}\left[w_i(\mathbf{X})-h(a_i,\theta_i)\right],
\end{equation}
where, for concreteness, the effort cost is defined as $h(a_i,\theta_i) = \frac{a_i^2}{2\theta_i}$, so that it is increasing in the effort, and decreasing in the type. Finally, we assume that $X_i = a_i+Q(\theta_i)\eta_i $, where $Q(\theta_i)$ is a increasing function of the type. 

\section{One-Shot Contract}\label{sec:static}
Having a limited budget, the principal uses a Cournot-like payment function which discourages large aggregated efforts: 
\begin{equation}\label{eq:wcournot}
    w_i(\mathbf{X}) = X_i\left(\lambda-\sum_{i=1}^2X_i\right),
\end{equation}
where $\lambda$ is a design variable. Given the utility function \eqref{eq:U} and the payments \eqref{eq:wcournot}, the problem of the principal is given by
\begin{align}
\label{eq:principal} &\underset{\lambda}{\max}~U\\
\nonumber    
\text{s.t.}~\mathbb{E}[w_i(\mathbf{X})-h&(a_i^*,\theta_i)]\geq 0 ~\forall i\\
\nonumber    
a_i^* = \text{argmax}~& \mathbb{E}\left[w_i(\mathbf{X})-h(a_i,\theta_i)\right]~\forall i,
\end{align}
where the first constraint is the individual rationality (IR) constraint which reflects that the agents will only accept the contract if their expected utility is non-negative, and the second constraint is the incentive compatibility constraint 
which enforces that the efforts considered in the problem are such that they maximize the agents' utilities.

After solving his problem, the principal issues the contract to the agents. The payment \eqref{eq:wcournot} couples the outputs from both agents, and therefore competition is induced at the agents. Each agent solves the problem 
$   \underset{a_i}{\max}~u_i(X_i,X_{-i}),$
where we make explicit the dependence of agent $i$'s utility on his own output $X_i$ and the output of the other agent $X_{-i}$.

We begin by analyzing this duopoly in a scenario where the agents play their Cournot strategies, and then we establish conditions for which the agents can derive rent by colluding.

\paragraph{Cournot Duopoly}
We use backwards induction to find the equilibrium of this Cournot competition. 

\begin{proposition}\label{prop:cournot}
Given payments $w_i(\mathbf{X})$ as expressed in \eqref{eq:wcournot}, and in the absence of collusion,
the principal maximizes his own utility by choosing
\begin{equation}\label{eq:lambdastar}
    \lambda^* = \max \left\{\frac{(2\theta_1+1)(2\theta_2+1)-\theta_1\theta_2}{2(\theta_1+1)(\theta_2+1)},\underbar{$\lambda$}\right\}
\end{equation}
where $\underbar{$\lambda$}$ is the minimum value that guarantees individual rationality for all agents, and is given by
\begin{equation}
    \underbar{$\lambda$} = \max \left\{ \frac{\sqrt{2}Q(\theta_i)\sigma[(2\theta_1+1)(2\theta_2+1)-\theta_1\theta_2]}{(\theta_{-i}+1)\sqrt{\theta_i(2\theta_i+1)}} \right\}.
\end{equation}
In response, the optimal efforts of agents $i\in\{1,2\}$ is 
\begin{equation}\label{eq:aopt}
    a_i^* = \frac{\lambda^* \theta_i(\theta_{-i}+1)}{(2\theta_i+1)(2\theta_{-i}+1)-\theta_i\theta_{-i}},
\end{equation}
where $\theta_{-i}$ is the other agent in the duopoly.
\end{proposition}
\begin{proof}
For a given $\lambda$, the utility of agent $i=1$ is
\begin{align}
    u_1
    &= a_1(\lambda - a_1 -a_2)- Q(\theta_1)^2\sigma^2 -a_1^2/2\theta_1,
\end{align}
where we used the fact that the outputs $X_i$ are independent, $\mathbb{E}[X_i]=a_i~\forall i$, and $\mathbb{E}[X_i^2]=a_i^2+Q(\theta_1)^2\sigma^2~\forall i$. This function is concave in the decision variable $a_1$, and thus we use the first order condition to derive the equilibrium effort
\begin{equation}\label{eq:a1star}
    \frac{\partial u_1}{\partial a_1}=\lambda-2a_1-a_2-\frac{a_1}{\theta_1}=0\implies a_1^*=\frac{\theta_1(\lambda-a_2^*)}{2\theta_1+1}.
\end{equation}
Proceeding similarly for agent $i=2$, we can find an expression for the equilibrium effort $a_2^*$ as a function of $a_1^*$, then substitute it in \eqref{eq:a1star} to find the expression \eqref{eq:aopt}. For individual rationality to hold, the agents must have a non-negative utility. For the optimal efforts derived, this condition reduces to
\begin{equation}\label{eq:uigeq0}
    u_i=\frac{\lambda^2\theta_i(\theta_{-i}+1)^2(2\theta_i+1)}{2\left[(2\theta_i+1)(2\theta_{-i}+1)-\theta_i\theta_{-i}\right]^2}- Q(\theta_i)^2\sigma^2\geq 0~\forall i,
\end{equation}
which establishes a lower bound on $\lambda$.
%
With these 
optimal efforts, 
the principal's utility can be written as
\begin{align*}
    U 
    = \lambda &g(\theta_1,\theta_2)-\lambda^2(g(\theta_1,\theta_2)-g(\theta_1,\theta_2)^2)\\
    &+ [Q(\theta_1)^2+Q(\theta_2)^2]\sigma^2,
\end{align*}
where $g(\theta_1,\theta_2)=(a_1^*+a_2^*)/\lambda$ and is given by
\begin{equation}
    g(\theta_1,\theta_2) = \frac{\theta_1(\theta_2+1)+\theta_2(\theta_1+1)}{(2\theta_1+1)(2\theta_2+1)-\theta_1\theta_2}.
\end{equation}
The principal's utility can be shown to be concave in $\lambda$. Thus, we can use the first order condition to find the optimal decision as in the first expression in \eqref{eq:lambdastar}. This value of $\lambda^*$ satisfies the individual rationality condition \eqref{eq:uigeq0} for all agents if
\begin{equation}
   Q(\theta_i)^2 \sigma ^2 \leq \frac{\theta_i(2\theta_i+1)}{8(\theta_i+1)^2}~\forall i.
\end{equation}
If the inequality above does not hold, then $\lambda^*$ must be equal to maximum of the lower bounds established through \eqref{eq:uigeq0}.
\end{proof}

The optimal effort \eqref{eq:aopt} can be shown to be increasing in agent $i$'s own type for a fixed $\theta_{-i}$, and thus agents with higher type values will exert more effort. Conversely, if agent $i$'s type is fixed, then his effort is decreasing with $\theta_{-i}$. This is due to the trade-off established through the payment function. 
We can also observe that agents with high output variance will require higher payments to participate in the contract, as their output may be very low even if a high effort is exerted.

\paragraph{Collusion Strategy}
Collusion becomes individually rational for both agents if the utility they derive by exerting the optimal effort \eqref{eq:aopt} can be surpassed by picking off-equilibrium efforts that will yield a surplus. 
In essence, the agents will 
collude if there exists a Pareto improvement 
that makes them both better off. Formally, this 
condition can be expressed as
\begin{equation}\label{eq:collusion}
    \exists~(\overline{a}_i,\overline{a}_{-i},P)  ~\text{s.t.}~ u_i(\overline{a}_i,\overline{a}_{-i})+P>u_i(a_i^*,a_{-i}^*)~\forall i,
\end{equation}
where $P$ is a possible side payment among the agents, which is positive (negative) for the agent receiving (making) it. 

Consider a scenario in which the two agents collude to exert efforts which, combined, yield a certain $\theta_i$-monopoly effort. Such effort is the optimal effort exerted by an agent of type $\theta_i$ when he is the only agent providing outputs to the principal, and is defined as $a^M_i = \lambda \theta_i/(2\theta_i+1)$.

For simplicity, we assume the following collusion strategies: 
\begin{itemize}
    \item if agents are of the same type, they equally split the amount of effort exerted to reach monopoly level, and no side payment is performed, and
    \item if agents are of different types, the agent with higher type exerts all the monopoly-level effort and makes a side payment to the lower-type agent, who exerts no effort.
\end{itemize} 

The next result identifies conditions for the agents to have an incentive to collude.

\begin{theorem}\label{th:collusion}
If the agents are of the same type, that is $\theta_1=\theta_2=\theta$, then collusion to mimic the $\theta_1$-monopoly effort by choosing the efforts $ \overline{a}_1=\overline{a}_2=a_1^M/2$
is individual rational for both agents if $\theta>(1+\sqrt{17})/8\approx 0.64$. 
Further, if the agents are of different types $\theta_1=\theta_h>\theta_\ell=\theta_2$, and satisfy one of conditions
\begin{align}
    & \text{C1:}~\theta_\ell \geq 1+\sqrt{2} \\
    \text{or}~&\text{C2:} \begin{cases}
     0 \leq \theta_\ell < 1+\sqrt{2} ~ \text{and} ~\\ \theta_h \geq \frac{\sqrt{(\theta_\ell+1)(17\theta_\ell+9)}+3\theta_\ell+1}{2(\theta_\ell+2)}
    \end{cases}
\end{align}
then there exists a side payment $P$ from the $\theta_h$ agent to the $\theta_\ell$ agent that makes the colluding efforts $\overline{a}_1 = a_h^M, ~ \overline{a}_2=0$ to mimic a $\theta_h$-monopoly
individual rational.
\end{theorem}
\begin{proof}
Agents of the same type $\theta$ collude by exerting the same effort and without making a side payment. 
Then, their IR constraint \eqref{eq:collusion} is the same and can be written as
\begin{equation}
    \overline{a}_i(\lambda - \overline{a}_i -\overline{a}_{-i}) -\overline{a}_i^2/2\theta > a_i^*(\lambda - a_i^* -a_{-i}^*) -{a_i^*}^2/2\theta.
\end{equation}
Substituting the proposed collusion efforts $\overline{a}_1=\overline{a}_2=a_1^M/2$ and the Cournot strategies \eqref{eq:aopt}, we can find the inequality 
\begin{align}
    \frac{\lambda^2\theta(4\theta+3)}{8(2\theta+1)^2} &> \frac{\lambda^2\theta(\theta+1)^2(2\theta+1)}{2\left[(2\theta+1)^2-\theta^2\right]^2}\\
    \implies 4\theta^5+11\theta^4+8\theta^3 &> 2\theta^2+4\theta+1,
\end{align}
which holds for $\theta>(1+\sqrt{17})/8$.

If agents are of different types $\theta_\ell$ and $\theta_h>\theta_\ell$, let the side payment from agent $\theta_h$ to agent $\theta_\ell$ be a fraction of agent $\theta_h$'s expected payment, that is $|P| = \alpha \mathbb{E}[w_h(\mathbf{X})]$, where $0\leq \alpha \leq 1$. 
Define $D = 2\left[(2\theta_h+1)(2\theta_\ell+1)-\theta_h\theta_\ell\right]^2,$ $R_h = \frac{Q(\theta_h)^2\sigma^2}{\lambda^2\theta_h},$  and $R_\ell = \frac{Q(\theta_\ell)^2\sigma^2}{\lambda^2\theta_h}.$
For the collusion efforts $\overline{a}_1 = a_h^M$ and $\overline{a}_2=0$, the constraint \eqref{eq:collusion} for agent $\theta_h$ reduces to
\begin{equation}
    \frac{(\theta_h+1/2)-\alpha(\theta_h+1)}{(2\theta_h+1)^2}+\alpha R_h > \frac{(\theta_\ell+1)^2(2\theta_h+1)}{D},
\end{equation}
from which we can find the following upper bound on $\alpha$
\begin{equation}
    \alpha < \frac{(2\theta_h+1)D/2-(\theta_\ell+1)^2(2\theta_h+1)^3}{D\left[(\theta_h+1)-R_h(2\theta_h+1)^2\right]}.
\end{equation}
For agent $\theta_\ell$, the individual rationality constraint \eqref{eq:collusion} becomes
\begin{equation}
    \frac{\alpha (\theta_h+1)}{(2\theta_h+1)^2}-\alpha R_h > \frac{\theta_\ell (\theta_h+1)^2(2\theta_\ell+1)}{D\theta_h}-R_\ell,
\end{equation}
leading to the lower bound on the side payment fraction 
\begin{equation}\label{eq:alphalb}
    \alpha > \frac{[\theta_\ell(\theta_h+1)^2(2\theta_\ell+1)-D\theta_hR_\ell](2\theta_h+1)^2}{D\theta_h\left[(\theta_h+1)-R_h(2\theta_h+1)^2\right]}.
\end{equation}
For this collusion to be individual rational for both agents, the bounds established for $\alpha$ must yield a non-empty interval. The conditions presented in Theorem~\ref{th:collusion} are sufficient, and can be found for the case where agent $\theta_\ell$'s output has zero variance, so that $R_\ell=0$ and the lower bound \eqref{eq:alphalb} becomes higher. 
\end{proof}

When agents are of the same type, the lower bound on $\theta$ presented in Theorem~\ref{th:collusion} indicates that, when the type value of the agents is too low, they cannot improve their  utilities by colluding to mimic the $\theta_1$-monopoly effort. Since the effort is increasing with type, this could be seen as a limiting situation in which the agents cannot exert an effort that is high enough to drive prices significantly. Similarly, when agents are of different types, collusion becomes rational when the agent with the lowest type exceeds a certain threshold; otherwise, the high-type agent must have a high enough type value for collusion to occur. We note that, under a static contract, the agents satisfying these conditions would still derive rent by colluding if they interact finitely many times with the principal.

In the following result, we relax the assumption that agents of the same type would mimic their own type's monopoly to provide conditions for which they can extract extra rent by choosing a combined effort that imitates a different monopoly.

\begin{corollary}\label{cor:diffmonopoly}
Let agents of the same type $\theta$ equally split the effort exerted, and not make side payments when colluding. Then, for $\theta<\hat{\theta}< \theta(4\theta+3)$, collusion to mimic a $\hat{\theta}$-monopoly yields a higher utility to both agents than if they mimic their own type's monopoly. 
\end{corollary}
\begin{proof}
We check if it is individual rational for both agents to mimic a $\hat{\theta}$-monopoly, for any $\hat{\theta}\neq\theta$. 
For this case, the agents have the same individual rationality constraint, given by
\begin{equation}
    \frac{\lambda^2\hat{\theta}(4\theta\hat{\theta}+4\theta-\hat{\theta})}{8\theta(2\hat{\theta}+1)^2}>\frac{\lambda^2\theta(4\theta+3)}{8(2\theta+1)^2},
\end{equation}
which compares the agent's utility in a collusion at a $\hat{\theta}$-monopoly and at a $\theta$-monopoly. This inequality can be evaluated to show that it only holds if $\theta<\hat{\theta}\leq \theta(4\theta+3)$.
\end{proof}

Corollary~\ref{cor:diffmonopoly} shows that it is not individual rational for agents of the same type $\theta$ to jointly exert a $\hat{\theta}$-monopoly effort for any $\hat{\theta} < \theta$. In this case, the payments received would be too low for extra rent to be derived. On the other hand, pretending to be a monopoly of a type $\hat{\theta}>\theta$ can be beneficial to both agents, as long as this higher type satisfies the upper bound presented. For any other $\hat{\theta}> \theta(4\theta+3)$, the effort cost would dominate,  
nullifying any possible increase in the payment received.

\section{Dynamic Contract}\label{sec:dynamic}
If the principal and the agents interact continuously in an infinite horizon setting,  the principal is able to perform a noisy verification of the agents' outputs, so as to learn about their behavior and adjust the contract if necessary. 
For notational simplicity, in what follows, let
\begin{equation}
    \hat{\lambda}_i = \frac{\lambda^*(\theta_{-i}+1)}{(2\theta_1+1)(2\theta_2+1)-\theta_1\theta_2}.
\end{equation}
 
In the dynamic contract, the principal collects a sequence of $n$ outputs $\mathcal{X}_i:=\{X_i^1,...,X_i^n\}$ from each agent $i$, and performs the following hypothesis test:
\begin{align}
    \delta^i(\mathcal{X}_i) &= \begin{cases}
    H_0^i& \text{if}~ \big|\overline{X}_i-\hat{\lambda}_i\theta_i\big|<\epsilon 
    \\
    H_1^i& \text{otherwise}
    \end{cases}
    \raisetag{1\normalbaselineskip}
\end{align}
where $\overline{X}_i = \frac{1}{n}\sum_{k=1}^n X_i^k$
is the sample mean. The hypothesis $H_0^i$ indicates that agent $i$ is playing the Cournot equilibrium, whereas $H_1^i$ flags that agent $i$ is deviating from this strategy, which is indicative that collusion is underway. If both agents are flagged, the principal updates the payment function so as to penalize them for the remainder of the time, disincentivizing collusion. This verification process is repeated every $n$ time instants, and the principal utilizes the latest $n$ samples to perform the hypothesis test, as detailed in Algorithm~\ref{alg:contract}.

\begin{algorithm}
\caption{Dynamic Contract}
\label{alg:contract}
\begin{algorithmic}[1]
\State $k\gets 1$
\For{$k\leq n$}
\State Principal issues contract \eqref{eq:wcournot} with $\lambda$ as in Eq. \eqref{eq:lambdastar}
\If {Agent $i$ accepts contract and exerts effort $a_i$}
\State $X_i^k \gets a_i+Q(\theta_i)\eta_i$
\EndIf
\State $k\gets k+1$
\EndFor
\Repeat
\If{Any agent rejects any contract for $k\in[k-n+1,k]$}
\State $\lambda \gets 0$
\ElsIf {$\delta^i(X_i^{k-n+1},...,X_i^k)=H_1^i~\forall i$}
\State $\lambda \gets 0$ 
\Else
\State $\lambda$ unchanged
\EndIf
\State $t\gets 1$
\For{$t\leq n$}
\State Principal issues contract \eqref{eq:wcournot} with most recent $\lambda$
\If {Agent $i$ accepts contract}
\State $X_i^t \gets a_i+Q(\theta_i)\eta_i$
\EndIf
\State $k\gets k+1$
\State $t\gets t+1$
\EndFor
\Until {$k\leq K$}
\end{algorithmic}
\end{algorithm}

\begin{theorem}
For the dynamic contract in Algorithm~\ref{alg:contract}, if the time window during which data is collected is $n\in o(K)$, then the principal can detect if the agents are playing off-equilibrium strategies almost surely as $K\to\infty$.
\end{theorem}
\begin{proof}
When agent $i$ plays Cournot, the probability of false alarm can be bounded as
\begin{equation*}
    \text{Prob}\left\{\left|\overline{X}_i-\hat{\lambda}_i\theta_i\right|\geq\epsilon\right\}\leq 2\exp{\left(-\frac{2n\epsilon^2}{(\overline{\eta}_i-\underline{\eta}_i)^2}\right)}
\end{equation*}
using Hoeffding's inequality.
This probability falls exponentially to zero as $n\to\infty$. If the agents collude, the probability of missed detection is given by $\text{Prob}\{|\overline{X}_i-\hat{\lambda}_i\theta_i|<\epsilon\}$. In this case, the sample mean $\overline{X}_i\to a_i$ as $n\to\infty$ by the law of large numbers. Since the efforts exerted by the colluding agents are different from their Cournot equilibrium, $\text{Prob}\{|\overline{X}_i-\hat{\lambda}_i\theta_i|<\epsilon\}\to 0$. Thus, as $K\to \infty$, the probability of incorrect detection falls to zero as long as $n\in o(K)$.
\end{proof}

The last result indicates that the principal can eliminate the incentive for collusion in the long run by dynamically updating the contract according to the outputs received. 

\section{Conclusion}\label{sec:conclusion}
We formulated and analyzed a contract design problem for a principal who seeks to incentivize agents to exert costly effort to perform a task. We considered agents to be coupled in their compensation through a Cournot-like payment, and allowed them to collude. For a static scenario, we showed that agents can extract rent by colluding. 
However, if the principal interacts infinitely often with the agents,  a dynamic contract can be implemented by the principal to use historical data as a way to detect and disincnetivize such collusion. 
\ifCLASSOPTIONcaptionsoff
  \newpage
\fi



\bibliographystyle{IEEEtran}
\bibliography{references}

\begin{thebibliography}{10}
\providecommand{\url}[1]{#1}
\csname url@samestyle\endcsname
\providecommand{\newblock}{\relax}
\providecommand{\bibinfo}[2]{#2}
\providecommand{\BIBentrySTDinterwordspacing}{\spaceskip=0pt\relax}
\providecommand{\BIBentryALTinterwordstretchfactor}{4}
\providecommand{\BIBentryALTinterwordspacing}{\spaceskip=\fontdimen2\font plus
\BIBentryALTinterwordstretchfactor\fontdimen3\font minus
  \fontdimen4\font\relax}
\providecommand{\BIBforeignlanguage}[2]{{%
\expandafter\ifx\csname l@#1\endcsname\relax
\typeout{** WARNING: IEEEtran.bst: No hyphenation pattern has been}%
\typeout{** loaded for the language `#1'. Using the pattern for}%
\typeout{** the default language instead.}%
\else
\language=\csname l@#1\endcsname
\fi
#2}}
\providecommand{\BIBdecl}{\relax}
\BIBdecl

\bibitem{gao}
H.~{Gao}, C.~H. {Liu}, W.~{Wang}, J.~{Zhao}, Z.~{Song}, X.~{Su},
  J.~{Crowcroft}, and K.~K. {Leung}, ``A survey of incentive mechanisms for
  participatory sensing,'' \emph{IEEE Communications Surveys Tutorials},
  vol.~17, no.~2, pp. 918--943, 2015.

\bibitem{restuccia}
\BIBentryALTinterwordspacing
F.~Restuccia, S.~K. Das, and J.~Payton, ``Incentive mechanisms for
  participatory sensing: Survey and research challenges,'' \emph{ACM Trans.
  Sen. Netw.}, vol.~12, no.~2, Apr. 2016. [Online]. Available:
  \url{https://doi-org.proxy.library.nd.edu/10.1145/2888398}
\BIBentrySTDinterwordspacing

\bibitem{garnaev}
A.~{Garnaev}, A.~{Petropulu}, W.~{Trappe}, and H.~V. {Poor}, ``A power control
  game with uncertainty on the type of the jammer,'' in \emph{2019 IEEE Global
  Conference on Signal and Information Processing (GlobalSIP)}, 2019, pp. 1--5.

\bibitem{mukherjee}
A.~{Mukherjee} and A.~L. {Swindlehurst}, ``Jamming games in the mimo wiretap
  channel with an active eavesdropper,'' \emph{IEEE Transactions on Signal
  Processing}, vol.~61, no.~1, pp. 82--91, 2013.

\bibitem{holmstrom1999managerial}
B.~Holmstr{\"o}m, ``Managerial incentive problems: A dynamic perspective,''
  \emph{The review of Economic studies}, vol.~66, no.~1, pp. 169--182, 1999.

\bibitem{prat2014dynamic}
J.~Prat and B.~Jovanovic, ``Dynamic contracts when the agent's quality is
  unknown,'' \emph{Theoretical Economics}, vol.~9, no.~3, pp. 865--914, 2014.

\bibitem{dobakhshari2}
D.~G. {Dobakhshari}, P.~{Naghizadeh}, M.~{Liu}, and V.~{Gupta}, ``A
  reputation-based contract for repeated crowdsensing with costly
  verification,'' \emph{IEEE Transactions on Signal Processing}, vol.~67,
  no.~23, pp. 6092--6104, 2019.

\bibitem{farokhi_dynamic}
F.~{Farokhi}, A.~M.~H. {Teixeira}, and C.~{Langbort}, ``Estimation with
  strategic sensors,'' \emph{IEEE Transactions on Automatic Control}, vol.~62,
  no.~2, pp. 724--739, 2017.

\bibitem{venkitasubramaniam}
P.~{Venkitasubramaniam} and V.~{Gupta}, ``Data-driven contract design,'' in
  \emph{2019 American Control Conference (ACC)}, 2019, pp. 2283--2288.

\bibitem{westenbroek}
T.~{Westenbroek}, R.~{Dong}, L.~J. {Ratliff}, and S.~S. {Sastry}, ``Competitive
  statistical estimation with strategic data sources,'' \emph{IEEE Transactions
  on Automatic Control}, vol.~65, no.~4, pp. 1537--1551, 2020.

\bibitem{dobakhshari}
D.~G. {Dobakhshari} and V.~{Gupta}, ``A contract design approach for phantom
  demand response,'' \emph{IEEE Transactions on Automatic Control}, vol.~64,
  no.~5, pp. 1974--1988, 2019.

\bibitem{farokhi}
F.~{Farokhi}, I.~{Shames}, and M.~{Cantoni}, ``Promoting truthful behavior in
  participatory-sensing mechanisms,'' \emph{IEEE Signal Processing Letters},
  vol.~22, no.~10, pp. 1538--1542, 2015.

\end{thebibliography}

%








\end{document}